\documentclass[journal]{IEEEtran} 
\ifCLASSINFOpdf
\else
\fi

\usepackage{cite}
\usepackage{stfloats}
\ifCLASSINFOpdf
   \usepackage[pdftex]{graphicx}
   \graphicspath{{../pdf/}{../jpeg/}}
   \DeclareGraphicsExtensions{.pdf,.jpeg,.png}
\else
   \usepackage[dvips]{graphicx}
   \graphicspath{{../eps/}}
   \DeclareGraphicsExtensions{.eps}
\fi

\usepackage{amsthm}
\usepackage[cmex10]{amsmath}
\usepackage{mathrsfs,color,amssymb}
\usepackage{amsfonts}

\usepackage{algorithm}
\usepackage{algorithmic}
\allowdisplaybreaks

\ifCLASSOPTIONcompsoc
\else
  \usepackage[caption=false,font=footnotesize]{subfig}
\fi

\usepackage{hyperref}
\usepackage{graphicx}
\usepackage{epstopdf}

\newtheorem{remark}{\bf{Remark}}

\newtheorem{lemma}{\bf{Lemma}}

\begin{document}

%

\title{\LARGE{Energy Efficient Beamforming for Massive MIMO Public Channel}
}
%
%
%
\author{Cheng~Zhang,~\IEEEmembership{Student Member,~IEEE,}
        Yongming~Huang,~\IEEEmembership{Member,~IEEE,}
        Yindi~Jing,~\IEEEmembership{Member,~IEEE,}
        and~~Luxi~Yang,~\IEEEmembership{Member,~IEEE}
        \vspace{-12mm}
\thanks{
This work has been accepted by IEEE Transactions on Vehicular Technology.


C. Zhang, Y. Huang, and L. Yang are with the School of Information Science and Engineering, Southeast University, Nanjing 210096, P. R. China (email:
{zhangcheng1988, huangym, lxyang}@seu.edu.cn).

Y. Jing is with the Department of Electrical and Computer Engineering, University of Alberta, Edmonton, Canada, T6G 2V4 (email: yindi@ualberta.ca).}
}

%
%

\markboth{}
{Shell \MakeLowercase{\textit{et al.}}: Bare Demo of IEEEtran.cls for Journals}
%



\maketitle

\begin{abstract}
For massive MIMO public channel with any sector size in either microwave or millimeter wave (mmwave) band, this paper studies the beamforming design to minimize the transmit power while guaranteeing the quality of service (QoS) for randomly deployed users. First the ideal beampattern is derived via Parseval Identity, based on which a beamforming design problem is formulated to minimize the gap with the idea beampattern. The problem is transformable to a multiconvex one and an iterative optimization algorithm is used to obtain the full-digital beamformer. In addition, with the help of same beampattern theorem, the power amplifier (PA) efficiency of the beamformer is improved with unchanged beampattern. Finally, the practical hybrid implementation is obtained that achieves the full-digital beamformer solution. Simulations verify the advantages of the proposed scheme over existing ones.
\end{abstract}

\begin{IEEEkeywords}
Massive MIMO, public channel, energy efficient, beamforming, beampattern.
\end{IEEEkeywords}\vspace{0mm}

%
\IEEEpeerreviewmaketitle

\vspace{-5mm}
 \section{Introduction}
In massive MIMO systems,
multi-user beamforming in dedicated channels has been investigated intensively \cite{Hoydis,chengzhang}. But designs for public channels are relatively limited. Public channels play crucial roles in broadcasting essential synchronization, reference, and control signals \cite{xinmeng1}. Public channel transmissions can be divided into two categories \cite{Meng_Omnidirectional}: closed-loop and open-loop. For closed-loop approaches, instantaneous or statistic channel state information (CSI) is assumed at the BS and the transmission is optimized to be adaptive to the CSI, e.g., choosing the optimal precoding matrix to maximize the worst-case receiving signal-to-noise ratio (SNR) \cite{Sidiropoulos-Transmit,Lozano-long}. Open-loop approaches assume no CSI and the BS broadcasts common information blindly to users. For some applications, quality CSI is difficult to obtain (e.g., users may be silent for a long time \cite{Larsson_Joint}). Meanwhile, when the users' covariance matrices are largely different, CSI-based designs can have higher complexity and little performance superiority. Thus, this work considers open-loop schemes \cite{xuezhiyang,xinmeng1}.

Several beamforming designs have been proposed in existing literature for public channel communications in conventional MIMO systems, e.g., global search, randomization, and communication standard based schemes (see \cite{xuezhiyang} and references therein). Unfortunately, these methods may be incompatible with massive MIMO systems \cite{xinmeng1}. A new design for massive MIMO public channel can be found in \cite{xinmeng1}. Based on a large-scale approximation of the channel correlation matrix, Zadoff-Chu (ZC) sequence was utilized to design the beamformer which aims to make the signal powers at $M$ (which is the antenna number at the BS) discrete angles equal. But for large but finite antennas, $M$ discrete angles is insufficient in representing the continuous sector, so the design has performance degradation for  unobserved angles. Moveover, the ZC scheme is dedicated for the standard sector only, e.g., $[-90^\circ,90^\circ]$, with the antenna spacing ratio being 0.5. Its direct application to the adaptive sectorization scenario seems infeasible. In \cite{yeli}, by using the same beampattern theorem \cite{same_beam}, a broadbeam design was proposed where a reference beampattern is first determined, then the beamformer is obtained by solving a polynomial equation to approach the reference beampattern. Again, only the standard sector was considered and the polynomial coefficients tend to be inaccurate when the number of discrete angles for optimization is larger than $2M-1$.

Beside the widely adopted single-beam scheme mentioned above, another approach for massive MIMO public channel is precoding using space-time block
coding, which can have both omnidirectional coverage and high diversity performance \cite{Meng_Omnidirectional}. But generally speaking, this scheme has higher complexity due to the channel estimation overheard at the receiver side,  involved modulation designs, and the decoding. Another interesting direction is the joint dedicated channel and public channel design \cite{Larsson_Joint}. But this may not achieve better performance over   orthogonal access (where dedicated channel and public channel are given distinct time-frequency resources) for practical massive MIMO systems \cite{Larsson_Joint}. Thus, we focus on single-beam transmission design in public channel only.

For massive MIMO communications, two practical constraints are emphasized. First, due to the high cost of radio frequency (RF) hardware, especially for the millimeter wave (mmwave) band, having one high-resolution digital RF chain for each antenna tends to be infeasible. RF chains with low-resolution components \cite{Zhang_On}, \cite{Zhang_Performance} and the hybrid beamforming structure with less RF chains than antennas \cite{Spatially}, \cite{molisch1} are considered to be cost effective and have little performance loss compared with their full-digital counterpart. Second, the demand for high energy efficiency for 5G  makes both the constraint on the total transmit power and the power efficiency of each antenna important for massive MIMO systems \cite{jianjun}.

This paper is on the beamforming design of the low-complexity single-beam transmission for public channel in massive MIMO systems with the energy efficient hybrid structure and arbitrary sector size. The problem is formulated as the transmit power minimization under quality-of-service (QoS) constraint\footnote{The outage probability is used as the QoS metric. For massive MIMO public channel, short packet transmission is envisioned due to the demand for low complexity and continuously reliable decoding \cite{Karlsson_Broadcasting}. Such short packet usually only spans one channel realization, consequently, the outage performance is more relevant than the ergodic rate performance.} with additional consideration on the power amplifier (PA) efficiency.
In solving the beamforming design problem, first the ideal beampattern is derived, then a full-digital beamformer optimization is formulated as finding the beamformer whose beampattern has the minimum gap with the idea one.
The optimization is solved by an efficient iterative algorithm based on semi-definite programming (SDP). Subsequently, by drawing lessons from \cite{same_beam,yeli,molisch1,jianjun}, we utilize the same beampattern theorem to improve the beamformer's PA efficiency and decomposed the beamformer into the product of a baseband beamformer and an analog one to be applicable to the hybrid massive MIMO structure. Simulations show that the proposed design has superior performance over existing schemes.

The main difference of this work with existing ones is two-fold. First, the problem formulation is fundamentally  different. In existing work,  equal power at some discrete angles is required \cite{xinmeng1}, \cite{yeli}. Our formulation is from the perspective of the QoS guarantee, the energy efficiency\footnote{The term \textit{energy efficiency} is generally defined as the the ratio of spectrum efficiency to the consumed energy. In this paper, it refers to the transmit power minimization with guaranteed QoS.}, and the hybrid structure. Second, the proposed design procedure explores the ideal beampattern result and the FIR filter theorem. Compared with the ZC scheme \cite{xinmeng1}, the proposed one is not limited to the performance for $M$ discrete angles and is adaptive to the sector size. Compared with the broadbeam design \cite{yeli}, the proposed design does not require predefined reference beampattern whose design is non-trivial especially for arbitrary sector size. Moreover, the use of convex optimization tools enables the searching over a broader region compared with the polynomial equation solving in \cite{yeli}. Meanwhile, the strict requirement on the number and locations of the discrete angles for optimization can also be relaxed with the proposed design.

\vspace{-3mm}
\section{System Model}
Consider a massive MIMO BS\ with hybrid beamforming structure, where the number of antennas is $M$ and the number of RF chains is $N_{RF}$. Let $D$ be the antenna spacing ratio, $[\Theta _{1 },\Theta _{2 }]$ be the radiating interval of each antenna, and $\mathcal{A}\triangleq\left[ {{\Theta _{\min }},{\Theta _{\max }}} \right]$ be the sector of interest inside the antenna radiating interval. Assume that the ideal directional radiating antenna is employed at the BS, i.e., the signal from outside of $[\Theta _{1 },\Theta _{2 }]$ is completely suppressed by the antenna pattern \cite{xinmeng1}. To avoid angle ambiguity \cite{array}, it is assumed that
$\Theta _{1 }\ge-\pi/2$,   $\Theta _{2}\le \pi/2$  for $D \le 0.5$; and
$\Theta _{1 }\ge-\arcsin\left( {\frac{1}{2D}} \right)$,  $\Theta _{2}\le\arcsin\left( {\frac{1}{2D}} \right)$ for $D > 0.5$.
For the public channel, common messages are sent to all  users randomly distributed in the sector $\mathcal{A}$. No CSI is available at the BS.

The received signal at User $k$ is
\begin{equation}\label{eq_1}
\small
  y_k = \sqrt \rho  {{\bf{h}}_k^H}{\bf{w}}x + z_k,
\end{equation}
where ${\bf h}_k\in {{\mathbb C}^{M \times 1}}$ is the channel vector, ${\bf{w}}\in {{\mathbb C}^{M \times 1}} $ is the beamformer normalized as ${\bf{w}}^H{\bf{w}}=1$,  $\rho$ is the average transmit power, $x\sim \mathcal {CN}(0,1)$ is the common signal symbol, and $z_k \sim \mathcal{ CN}\left( {0,1} \right)$ is the local noise. The notation $\mathcal{ CN}\left( {0,1} \right)$ represents the circularly symmetric complex Gaussian distribution with zero-mean and unit-variance and $x\sim \mathcal{ CN}\left( {0,1} \right)$ means that $x$ follows such distribution. With hybrid structure, ${\bf{w}}= {{\bf{W}}_{RF}}{{\bf{w}}_{BB}}$, where ${{\bf{W}}_{RF}} \in {{\mathbb C}^{M \times {N_{RF}}}}$ is the analog beamforming matrix with the constraint $\left|[{\bf{W}}_{RF}]_{i,j}\right|=1$ and ${{\bf{w}}_{BB}} \in {{\mathbb C}^{{N_{RF}} \times 1}}$ is the baseband beamformer, where $[{\bf{W}}_{RF}]_{i,j}$ is the $(i,j)$-th entry of ${\bf{W}}_{RF}$.

The spatially correlated channel is considered, i.e.,
\begin{equation}
\label{Eq2}
{{\bf{h}}_k} = {\bf{R}}_k^{1/2}{\bf{h}}_k^{iid},
\end{equation}
where ${\bf h}_k^{iid}\sim{{\cal {CN}}({\boldsymbol 0,\boldsymbol I_{M}})}$ is the fast fading channel component and ${\bf R}_k\in {{\mathbb C}^{M \times M}}$ is the channel covariance matrix. For the one-ring scattering model under a far-field assumption, the channel covariance matrix  ${\bf R}_k$ can be modeled as
\begin{equation}
\label{eqn6}
{\bf{R}}_k = \int_{{{\Theta _{1 }}}}^{{{\Theta _{2}}} } {f_k(\theta ){\boldsymbol \alpha }(\theta ){{\boldsymbol{\alpha }}^H}(\theta )} d\theta ,
\end{equation} where $f_k(\theta)$ is the power azimuth spectrum (PAS) \cite{xinmeng1} which indicates the joint effect of each antenna's gain pattern and the channel scattering distribution and ${\boldsymbol \alpha }(\theta )$ is the array vector for the physical angle $\theta$. Assume uniform linear array at the BS\footnote{Our results can be straightforwardly extended to other antenna array topologies such as uniform planar arrays or uniform circular arrays.}, we have
${\boldsymbol{\alpha }}\left( \theta  \right) = {\left[ {1,...,{e^{ - j2\pi D\sin \left( \theta  \right)\left( {M - 1} \right)}}} \right]^T}.$

Denote the set of all possible user PAS is as follows:
\begin{eqnarray}
\mathcal{F}\triangleq\left\{f(\theta)\bigg{|}
\begin{array}{ll} \int_{\theta^{min} }^{\theta^{max}} \hspace{-1mm}f(\theta )d\theta = 1,\\
f(\theta )=0 \text{ for } \theta\hspace{-0.5mm}\notin\hspace{-0.5mm} \left[ {\theta^{min} ,\theta^{max} } \right] \hspace{-0.5mm}\subseteq \hspace{-0.5mm}
\mathcal{A} \end{array}
 \right\}.
 \label{PAS}
\end{eqnarray}
In the constraint in (\ref{PAS}), $\left[ \theta^{min}, \theta^{max}  \right] \subseteq \mathcal{A}$ is the channel angular spread (AS) interval and the first part is for the normalization. Users with different PAS have different channel covariance matrices, thus different channel distribution. Notice that for User k, when its PAS $f_k(\theta)\in \mathcal{F}$, the condition in (\ref{PAS}) implies ${\rm tr}\{{\bf R}_k\}={M}$. While this normalization does not incorporate the large-scaling fading of user channel, it can be seen as considering users with the worst large-scale fading. Further, by properly setting $f_k(\theta)$, many channel types can be modeled. For example, if $f_k(\theta)$ is set to be a Dirac delta function, i.e., $\left[ \theta_k^{min}, \theta_k^{max}\right]$ is an extremely small interval, the corresponding channel is a single-path one typical in the mmWave band.

\vspace{-2.5mm}
\section{Problem Formulation and Solution Procedure}
In this section, the beamforming design problem is first formulated. Then the solution procedure is provided.
\vspace{-3.5mm}
\subsection{Problem Formulation}
For the public channel transmission, the beamforming design needs to guarantee the worst performance within the sector. The outage probability, $P_{out}$, is chosen to be the performance measure. In addition, the energy efficiency is important especially from the perspective of communication carriers. Therefore, the energy efficient hybrid beamforming design problem for massive MIMO public channel is formulated as
\begin{equation*}\label{eq16}
\small
\begin{aligned}
\textbf{P1}: &\underset{{\bf{w}}}{\min}
& & \rho \\
& \text{s.t.}
& & \hspace{-2mm}\mathop {\max }_{f_k(\theta)\in\mathcal{F}}\limits\{P_{out}\left(\rho, R, {\bf w}, f_k(\theta) \right)\} \le {\bar P}_{out}, \forall k, \hspace{20pt}\text{(C.1)}\\
&&& {\bf{w}} = {{\bf{W}}_{RF}}{{\bf{w}}_{BB}} \text{ and } |{[{\bf W}_{RF}]}_{i,j}|=1, \forall i,j, \text{\hspace{16pt}(C.2)}\\
&&& {\bf w}^H{\bf w}=1.\text{\hspace{140pt}(C.3)}
\end{aligned}
\end{equation*}where ${{\bar P}_{out}}$ is the given maximum acceptable outage probability and $R$ is the minimum rate. In P1, the transmit power is the optimization metric and the rate constraint $R$ functions through the outage probability condition in (C.1). Define the \textit{beampattern} at the angle  $\theta$ as
\begin{equation}\label{10_beampattern}
\small
g\left( \theta  \right) \triangleq {{\bf{w}}^H}{\boldsymbol{\alpha }}\left( \theta  \right){{\boldsymbol{\alpha }}^H}\left( \theta  \right){\bf{w}}.
\end{equation}
The outage probability can be calculated as
{\small
\setlength{\arraycolsep}{1pt}
\begin{eqnarray}\label{eq_17_1}
&&\hspace{-0.3cm}P_{out}\left(\rho, R, {\bf w}, f_k(\theta) \right) \nonumber \\
&&\hspace{-0.3cm}= {\rm Pr}\left\{ {\log_2\left( 1 + \rho |{\bf h}_k^H{\bf w}|^2 \right) < R} \right\} \nonumber \\
   &&\hspace{-0.3cm}\mathop {\rm{ = }}\limits^{\left( a \right)} {\rm{Pr}}\left\{ {{{\left| {h_k^{iid}} \right|}^2} \int_{{\Theta_{min}}}^{{\Theta_{max}}} {{f}_k\left( \theta  \right)g\left( \theta  \right)d\theta }  < \frac{{{2^R} - 1}}{\rho }} \right\} \nonumber \\
   &&\hspace{-0.3cm}\mathop {\rm{ = }}\limits^{\left( b \right)}\hspace{-0.1cm} {\rm{Pr}}\hspace{-0.1cm}\left\{\hspace{-0.1cm} {{{\left| {h_k^{iid}} \right|}^2} \hspace{-0.2cm}\int_{{\sin(\Theta_{min})}}^{{\sin(\Theta_{max})}}\hspace{-0.2cm} {\frac{{f}_k\left( \arcsin(x)  \right)g\left( \arcsin(x)  \right)}{\sqrt{1-x^2}}dx }  \hspace{-0.1cm}<\hspace{-0.1cm} \frac{{{2^R} \hspace{-0.1cm}- \hspace{-0.1cm}1}}{\rho }} \hspace{-0.1cm}\right\}\hspace{-0.1cm},
\end{eqnarray}
\setlength{\arraycolsep}{5pt}}where $h_k^{iid}\sim{{\cal {CN}}({0,1})}$. (a) follows from \eqref{Eq2}, the eigen-decomposition of ${{\bf{R}}_k^{{1 \mathord{\left/
 {\vphantom {1 {2,H}}} \right.
 \kern-\nulldelimiterspace} {2,H}}}}{\bf{w}}{{\bf{w}}^H}{{\bf{R}}_k^{{1 \mathord{\left/
 {\vphantom {1 2}} \right.
 \kern-\nulldelimiterspace} 2}}}$ and its rank-1 property, and ${{{\bf{w}}^H}{\bf{R}}_k{{\bf w}}}= \int_{{\Theta_{min}}}^{{\Theta_{max}}} {{f}_k\left( \theta  \right)g\left( \theta  \right)d\theta }$. (b) follows from transforming the angle domain $\theta$ to the normalized \textit{spatial frequency} domain $\sin(\theta)$. Recall the normalization ${\rm tr}\{{\bf R}_k\}={M}$. When the worst case design is considered, $\rho$ can be understood as the normalized transmit power with respect to the multiplication of noise power and the worst case path-loss.


\subsubsection{Additional requirement}
Besides minimizing the transmit power, another important factor affecting the system energy efficiency is the PA efficiency which is quantified by the antenna-domain peak-to-average power ratio (PAPR)\footnote{While time-domain PAPR is typically used for orthogonal frequency division multiplexing (OFDM) systems, we focus on the antenna-domain PAPR which is important for massive MIMO public channel transmissions.} in this paper, i.e.,
\begin{equation}
\small
\delta=\frac{M\max_m{|w_m|^2}}{\|{\bf w}\|^2}.
\end{equation}
One way to incorporate PA efficiency is to directly set a constraint on the PAPR in P1, or formulate a double-objective problem by adding the minimization of $\delta$. But this makes the problem intractable
\cite{global}. Alternatively, we provide an indirect solution for this and make it as an additional consideration. More details are provided in Section \ref{3 steps} and \ref{PAPR reduction}.
\vspace{-4mm}
\subsection{Solution Procedure}\label{3 steps}
Due to the complicated condition in  (C.1) and the non-convexity of (C.2), to directly solve P1 is difficult. Instead, the following procedure with four steps is proposed:
\begin{enumerate}
\item Find the optimal beampattern $g^{\star}\left( \theta  \right)$ to minimize the  transmit power under the outage probability constraint only.
\item Find a full-digital beamformer, ${\bf w}_1^{\star}$, whose beampattern closely matches the optimal one found in Step 1.
\item  Find a beamformer ${\bf w}_2^{\star}$ with the same beampattern as that of ${\bf w}_1^{\star}$ but with a lower PAPR.
  \item Decompose ${\bf w}_2^{\star}$ into ${\bf W}_{RF}^{\star}$ and ${\bf w}_{BB}^{\star}$ to complete its hybrid implementation.
\end{enumerate}

In the following sections, details of each step will be provided. Step 3 is used to improve the PA efficiency of the object beamformer. If Steps 1, 2, and 4 can be solved precisely, especially for Step 2, i.e., a full-digital beamformer whose beampattern is the same as $g^{\star}\left( \theta  \right)$ can be found, the procedure will lead the optimal solution of P1. However, as will be explained in the following sections, Steps 1 and 4 can be solved precisely, but due to the FIR filter theorem, the ideal beampattern is unattainable given the finite dimension of ${\bf w}$. Thus, in Step 2, the goal is to find a beamformer whose beampattern closely matches the optimal one.

\vspace{-2mm}
\section{Solution Details}
This section provides the detailed formulations and solutions for the proposed four-step procedure.

\vspace{-3mm}
\subsection{Ideal Beampattern}\label{idealbeampattern}
The first step is to find the optimal beampattern that minimizes the transmit power under the outage probability constraint. It can be written as the following:
\begin{equation*}
\small
\begin{aligned}
\textbf{P2}:\ & \underset{g\left( \theta  \right)}{\text{minimize}}
& & \rho \\
& \hspace{5pt}\text{s.t.}
& & \hspace{-25pt}\mathop {\max }_{f_k(\theta)\in \mathcal{F}}\limits\{P_{out}\left(\rho, R, g\left( \theta  \right), f_k(\theta) \right) \} \le {\bar P}_{out}, \forall k, \hspace{12pt}\text{(C.1)}\\
&&& \hspace{-20pt}{\bf w}^H{\bf w}=1. \text{\hspace{143pt}(C.3)}\\
\end{aligned}
\end{equation*}The optimal beampattern $g^\star\left( \theta  \right)$ for P2, also called the \textit{ideal beampattern},  is given in the following lemma.
\begin{lemma}
\label{corollary 1}
The optimal beampattern for P2 is
\begin{equation}
\label{eq_15}\small
\begin{array}{l}
g^\star\left( \arcsin(x)  \right) = \left\{ {\begin{array}{*{20}{c}}
{\frac{1}{D(\sin(\Theta _{\max })-\sin(\Theta _{\min }))},x  \in \mathcal{\tilde{A}}}\\
{0,\hspace{86pt}x \in  \mathcal{\tilde{A}}^{-}}
\end{array}} \right.
\end{array},
\end{equation}
where $\xi^\star\triangleq(D(\sin(\Theta _{\max })-\sin(\Theta _{\min })))^{-1}$ is the beampattern value within the sector, $\mathcal{\tilde{A}} \triangleq \left[ {{\sin(\Theta _{\min })},{\sin(\Theta _{\max })}} \right]$ is the sector interval in the spatial frequency domain, and $\mathcal{\tilde{A}}^{-}\triangleq \left[ \sin(\Theta_1),\sin({\Theta _{\min }}) \right) \cup \left( \sin({{\Theta _{\max }}}),\sin(\Theta_2) \right]$ is the out-sector interval.
\end{lemma}
\begin{proof}
See Appendix A.
\end{proof}

\begin{remark} Ideal beampattern has been used in public channel beamforming design in \cite{xinmeng1}, \cite{yeli}.  But the proposed result in Lemma \ref{corollary 1} has distinctions in two aspects. First, beampatterns in \cite{xinmeng1}, \cite{yeli} are for the standard sector, e.g., $\mathcal{A}=[-90^\circ,90^\circ]$ and $D=0.5$ only; while our result is applicable for any sector size. Also, the beampattern in \cite{xinmeng1} is defined for only $M$ discrete angles. Second, beampatterns in \cite{xinmeng1}, \cite{yeli} are derived from the requirement of constant signal power at any discrete angles; while our result links the beampattern to direct communication objects.
\end{remark}

\begin{remark}
The result in Lemma \ref{corollary 1} shows that the ideal beampattern is independent of the outage parameters $R$ and ${\bar{P}}_{out}$ even though they appear in the Condition  (C.1) of P2. This is an interesting observation yet still reasonable. The objective in P2 is to minimize the transmit power $\rho$ with guaranteed outage performance of all possible users. Regardless of the specific value of $R$ and ${\bar{P}}_{out}$, the solution of P2 provides a flat beam strength to the targeted angle interval.
\end{remark}

\begin{remark}
When entries of ${\bf h}_k$ are independent and identically distributed (i.i.d.), the frequency domain PAS of User $k$, ${f_k(\arcsin(x))}/{\sqrt{1-x^2}}$, is  flat for $x \in \left[ {{\sin(\Theta _{\min })},{\sin(\Theta _{\max })}} \right]=\left[-{1}/{(2D)}, {1}/{(2D)}\right]$  \cite{xinmeng1,chengzhang}. From \eqref{eq_17_1} and \eqref{eq_14}, the ideal beampattern condition simplifies to $\int_{{\sin(\Theta _{\min })}}^{{\sin(\Theta _{\max })}} {g\left( \arcsin(x)  \right)} dx =1/D$, which is equivalent to ${\bf w}^H{\bf w}=1$. Compared with the case of spatially correlated channel, the demand for flatness of the beampattern is relaxed. This is because in this case all users share the same uniform scatterer distributed in the sector (in frequency domain) while for spatially correlated channel the scatterer distribution is different for different users.
\end{remark}

\vspace{-3mm}
\subsection{Ideal Beampattern Based Beamformer Design}\label{optimize w1}
The second step of the solution procedure is to find a full-digital beamformer, ${\bf w}_1^{\star}$, whose beampattern closely matches the optimal one. Recall the equivalence between the beampattern and the FIR filter response as shown in \eqref{equivalence_170323} of the proof for Lemma \ref{corollary 1}. Due to the FIR filter theorem \cite{Oppenheim}, there does not exist an $M$-dimensional vector ${\bf w}$ whose beampattern equals $g^\star\left( \theta  \right)$ in \eqref{eq_15}.
Specifically, the ripple in both $\mathcal{\tilde{A}}$ (pass band) and $\mathcal{\tilde{A}}^{-}$ are unavoidable. Moreover, the transition bandwidth $\Delta_T$ should be non-zero\footnote{Exception is for the full-pass filter, i.e., $\mathcal{\tilde{A}}=[\sin(\Theta_1),\sin(\Theta_2)]$, which does not need a transition band.} in practice. Thus $\mathcal{\tilde{A}}^{-}$ should be divided into two parts, i.e., the transition band and the stop band, in which the latter is $\mathcal{\tilde{A}}_s^{-}\triangleq\left[ \sin(\Theta_1),\sin({\Theta _{\min }})-\Delta_T \right) \cup \left( \sin({{\Theta _{\max }}})+\Delta_T,\sin(\Theta_2) \right]$.

Based on the FIR filter theorem, there is an  inherent tradeoff between the ripples in $\mathcal{\tilde{A}}$ and the size of transition bandwidth $\Delta _T$. Specifically, the ripple size in $\mathcal{\tilde{A}}$ is inversely proportional to $\Delta _T$ and the minimum beampattern decreases with increasing ripple for constant mean beampattern in $\mathcal{\tilde{A}}$. On the other hand,  due to the Parseval identity used in the proof for Lemma \ref{corollary 1}, a larger $\Delta _T$ may decrease the mean beampattern in $\mathcal{\tilde{A}}$, which decreases the minimum beampattern with constant ripple in $\mathcal{\tilde{A}}$.

By taking into consideration of the aforementioned phenomenon, in what follows, a formulation for the optimization of the full-digital beamformer is proposed to explore the tradeoff in the equivalent filter design problem.
\begin{equation*}\label{eq16}
\small
\begin{aligned}
\textbf{P3}: & \underset{{\bf{w}}_1}{\text{minimize}}
& & \sigma \\
& \hspace{23pt}\text{s.t.}
& & \left|{\left| {{{\bf{w}}_1^H} {\boldsymbol{\alpha}}(\arcsin (x))} \right|  - \sqrt{\xi^\star} } \right| \le \sigma ,x  \in \mathcal{\tilde{A}},\text{  ($\bar {\rm C}$.1)}\\
&&& {\left| {{{\bf{w}}_1^H} {\boldsymbol{\alpha}}(\arcsin (x))} \right|}  \le \sqrt{r_s}, x \in \mathcal{\tilde{A}}_s^{-},\text{\hspace{20pt}  ($\bar {\rm C}$.2)}\\
&&& 0< \Delta_T \le \Delta_T^{max}. \text{\hspace{90pt}  ($\bar {\rm C}$.3)}\\
&&& {{\bf{w}}_1^H}{\bf{w}}_1 \le 1. \text{\hspace{112pt}  ($\bar {\rm C}$.4)}
\end{aligned}
\end{equation*}
In this formulation, the square root of beampattern is used for simplicity. In ($\bar {\rm C}$.1), the gap between the beampattern of ${\bf{w}}_1$ and the ideal one ($\xi^\star$) in $\mathcal{\tilde{A}}$ should be lower than $\sigma$. In ($\bar {\rm C}$.2), to reduce possible inter-sector interference, a constraint on the maximum beampattern $r_s$ in $\mathcal{\tilde{A}}_s^{-}$ is required. Also from this perspective, we set another ($\bar {\rm C}$.3) on the maximum $\Delta_T$ (denoted as $\Delta_T^{max}$). In ($\bar {\rm C}$.4), ${{\bf{w}}_1^H}{\bf{w}}_1 \le 1$ is the convex relaxation version of ${{\bf{w}}_1^H}{\bf{w}}_1 =1$, which does not affect the optimization result since the equality can always be satisfied due to the structure of P3. Note that some finite discrete angles within $\mathcal{\tilde{A}}\cup\mathcal{\tilde{A}}_s^{-}$ should be selected for the optimization to be feasible.


\begin{remark}\label{remark5}
In \cite{yeli}, although the original object is to guarantee the performance of the whole continuous sector, only no more than $2M-1$ carefully designed angles are selected and the optimization is based on these discrete angles only.
Since P3 is solved numerically, angle discretization is still needed. But our method allows an arbitrary number of discrete angles and arbitrary selections of discrete angles. In fact, for high precision, a large number of discrete angles are considered to have a small angle spacing.
Meanwhile, the reference beampattern used for the beamforming design in \cite{yeli} is fixed as $[-90^\circ,90^\circ]$ only, and its application or extension to another sector size is non-trivial. Our formulation in P3 needs the ideal beampattern only and focuses on the key design requirements. Thus, the proposed design has an enlarged search space and results in the same or better performance. In addition, our scheme is applicable for an arbitrary sector size.
\end{remark}

\subsubsection{Algorithm Design} In P3, notice that  $\Delta_T$ only functions via $\mathcal{\tilde{A}}_s^{-}$. In solving P3, we first solve the problem for a given $\Delta_T\in(0,\Delta_T^{max}]$, then conduct a grid search over $\Delta_T$.

The difficulty lies in the non-convex constraint ($\bar {\rm C}$.1). By drawing lessons from \cite{semidefinite}, we can transform ($\bar {\rm C}$.1) to $|({\bf w}_{1,1}+{\bf w}_{1,2})^H{\boldsymbol{\alpha}}(\arcsin (x))|\le \sigma+\sqrt{\xi^\star}$
and $\max\{0,\sqrt{\xi^\star}-\sigma\}\le\sqrt{4{\rm Re}\{{\bf w}_{1,1}^H{\boldsymbol{\alpha}}(\arcsin (x))({\bf w}_{1,2}^H{\boldsymbol{\alpha}}(\arcsin (x)))^H\}}$ where ${\bf w}={\bf w}_{1,1}+{\bf w}_{1,2}$  \cite[Pro. 3.1]{semidefinite}. These two new constraints are multiconvex inequalities which is convex in ${\bf w}_{1,2}$ for given ${\bf w}_{1,1}$ and vice versa. Since ($\bar {\rm C}$.2) and ($\bar {\rm C}$.4) are convex, an iterative algorithm based on convex optimization \cite[Algorithm II]{semidefinite} can be used to efficiently solve P3 for an arbitrarily given $\Delta_T$. The initialization can be obtained by solving the standard semi-definite programming (SDP) problem via relaxing the inner absolute sign of ($\bar {\rm C}$.1). Due to the space limit and similarity, the pseudo-code and the proof for its convergence (refer to the end of \cite[Sec. III]{semidefinite}) are omitted here.

Based on the above analysis, the complete algorithm for Step 2 based on a grid search for the transition bandwidth $\Delta_T$ is summarized in Algorithm 1, where ${\rm d}\Delta_T$ is the step size of $\Delta_T$ and $g_{min}$ is the minimum beampattern in $\mathcal{\tilde{A}}$.
\vspace{-3mm}
\begin{algorithm}
\caption{Algorithm for obtaining ${\bf w}_1^\star$}
\label{alg1}
\begin{algorithmic}[1]
\STATE Input: $\xi^\star$, $\Delta_T^{max}$, ${\rm d}\Delta_T$, $r_s$;
\STATE Initialization: ${\bf w}_1^\star=\bf 0$; $ g_{min}=0$;
\FOR{$\Delta_T=[\Delta_T^{max},\Delta_T^{max}-{\rm d}\Delta_T,...,0)$}
\STATE Solve P3 with $\Delta_T$ and obtain $\bar {\bf w}_1$;
\STATE Denote the minimum beampattern in $\mathcal{\tilde{A}}$ as $\bar g_{min}$;
\IF{$\bar g_{min}>g_{min}$}
\STATE $g_{min}=\bar g_{min}$; ${\bf w}_1^\star=\bar {\bf w}_1$;
\ENDIF
\ENDFOR
\STATE Output: ${\bf w}_1^\star$.
\end{algorithmic}
\end{algorithm}

The complexity of Algorithm 1 depends on the step size and the convergence speed of the iterative optimization in Step 4. It is higher than that of the ZC scheme in \cite{xinmeng1} since the latter is directly based on the existing ZC sequence. As to the design in \cite{yeli}, if the reference beampattern is given, the complexity of obtaining the beamformer in general has lower complexity than that of solving P3 in our design. However, the only available reference beampattern design is for the standard sector. If the reference beampattern is not chosen appropriately, either higher complexity or worse performance will be resulted for the scheme in \cite{yeli} compared with our scheme. Our proposed method does not require such reference beampattern.

\subsection{PAPR Reduction and Hybrid implementation}\label{PAPR reduction}
The last two step of the solution procedure are to find the beamformer ${\bf w}_2^\star$ which has the same beampattern as that of ${\bf w}_1^\star$ obtained in Step 2 and has the lowest PAPR, and to decompose ${\bf w}_2^{\star}$  found in Step 3 into ${\bf W}_{RF}^{\star}$ and ${\bf w}_{BB}^{\star}$ for the practical hybrid implementation, respectively.
\subsubsection{PAPR Reduction}
By drawing lessons from \cite{same_beam}, ${\bf w}_2^\star$ can be found as follows.
First, we solve the polynomial equation: ${w}_{1,1}^{\star}+{w}_{1,2}^{\star}x+{w}_{1,3}^{\star}x^2+...+{w}_{1,M}^{\star}x^{M-1}=0$.
Without loss of geneality, the $M-1$ roots $x_1,x_2,...,x_{M-1}$ are assumed to be sorted so that $|x_i-1/x_i^*|$ is in decreasing order. The $2^{M-1}$ beamformers with the same beampattern as that of ${\bf w}_1^\star$ form the set $\mathcal{V}=\{{\bf v}|{v}_{1}+{v}_{2}x+{v}_{3}x^2+...+{v}_{M}x^{M-1}=\prod_{m=1}^{M-1}(x-a_m),a_m=x_m\text{ or }1/x_m^*, m=1,...,M-1\}$. Thus, we can search this set to find the one with the smallest PAPR.

If reduction on the search complexity is desired, the reduced set for the candidacy beamformers can be used: $\mathcal{\bar{V}}=\{{ \bar{\bf v}}|{\bar{v}}_{1}+{\bar{v}}_{2}x+{\bar{v}}_{3}x^2+...+{\bar{v}}_{M}x^{M-1}=\prod_{m=1}^{Q}(x-a_m)\prod_{m=Q+1}^{M-1}(x-a_m)\}$ with $a_m=x_m\text{ or }1/x_m^*, m=1,...,Q$ where $\{a_m,m=Q+1,...,M-1\}$ is predetermined randomly.

\subsubsection{Hybrid Implementation}
The geometric method in \cite{jianjun} can be used to construct the analog beamforming matrix ${\bf W}_{RF}^\star$ and the baseband beamformer ${\bf w}_{BB}^\star$ without performance loss for the hybrid structure with $N_{RF}>1$. First, to make the algorithm \cite[Algorithm 1]{jianjun} applicable, we give the following baseband design.
\begin{lemma}
By setting ${\bf w}_{BB}^\star=[b,b,...,b]^T$ with $b=\max\left\{|{w}_{2,i}^\star|/{N_{RF}},\forall i\right\}$, the sufficient condition of
triangle construction \cite[Theorem 1]{jianjun} is satisfied for ${\bf w}_2^\star={\bf W}_{RF}^\star{\bf w}_{BB}^\star$ with the constraint $|{[{\bf{W}}_{RF}^\star]}_{i,j}|=1,\forall i,j$.
\end{lemma}
\begin{proof}
For all $i$, by sorting $\{{w}_{2,i}^\star,{w}_{BB,1}^\star,...,{w}_{BB,N_{RF}}^\star\} $ according to their magnitude in decreasing order, we have $|{w}_{2,i}^\star|\leq\sum_{j=1}^{N_{RF}}|{w}_{BB,j}|=\max\left\{|{w}_{2,j}^\star|,\forall j\right\}$ or $|{w}_{BB,1}|\leq\sum_{j=2}^{N_{RF}}|{w}_{BB,j}|+|{w}_{2,i}^\star|$ which both satisfy the sufficient condition in \cite[Theorem 1]{jianjun} for successful triangle construction.
\end{proof}

With the given ${\bf w}_{BB}^\star$, ${\bf W}_{RF}^\star$ can be easily calculated by directly using \cite[Algorithm 1]{jianjun} (the details are omitted).

\vspace{-3mm}
\subsection{Discussions}
Instead of using the ideal beampattern based design to obtain the full-digital beamformer ${\bf w}_1^\star$, another method is to directly solve the following problem:
\begin{equation*}
\small
\begin{aligned}
\textbf{P4}: & \underset{{\bf{w}}_1}{\text{maximize}}
& & \eta \\
& \hspace{24pt}\text{s.t.}
& & \left| {{{\bf{w}}_1^H} {\boldsymbol{\alpha}}(\arcsin (x))} \right| \ge \eta ,x  \in \mathcal{\tilde{A}},\text{\hspace{35pt}  ($\tilde {\rm C}$.1)}\\
&&& {\left| {{{\bf{w}}_1^H} {\boldsymbol{\alpha}}(\arcsin (x))} \right|}  \le \sqrt{r_s}, x \in \mathcal{\tilde{A}}_s^{-},\text{\hspace{20pt}  ($\bar {\rm C}$.2)}\\
&&& 0< \Delta_T \le \Delta_T^{max}. \text{\hspace{90pt}  ($\bar {\rm C}$.3)}\\
&&& {{\bf{w}}_1^H}{\bf{w}}_1 \le 1. \text{\hspace{112pt}  ($\bar {\rm C}$.4)}
\end{aligned}
\end{equation*}
where the object along with ($\tilde {\rm C}$.1) is equivalent to the original one, i.e., minimizing $\rho$. P4 can be solved via a similar iterative optimization based algorithm to Algorithm \ref{alg1}.

\begin{remark}\label{direct_P2}
Compared with P3, P4 is a more direct formulation without using the ideal beampattern in \eqref{eq_15}.
If global optimal solutions of P4 and P3 can be found, directly solving P4 should
render a better beamformer. However, since both the constraint ($\bar {\rm C}$.1) in P3 and the constraint ($\tilde {\rm C}$.1) in P4 are non-convex, only local optimal solutions can be found and the initialization point for both algorithms is crucial for the performance \cite{semidefinite}. A direct solution with a random or naive initialization for P4 tends to fall in an unsatisfactory local optimal point often. On the other hand, the formulation in P3 is to find the closest match to the ideal beampattern with practical considerations. This approach, in some sense, has similar effect to taking a good initialization in the optimization.
\end{remark}


\vspace{-4mm}
\section{Simulation and Analysis}
In this section, the performance of the proposed scheme will be validated. We consider $D=0.5$, $M=64$, $[\Theta_{1}, \Theta_{2}]=[-90^\circ, 90^\circ]$, a $60^\circ$ sector where  $\Theta_{min}=-30^\circ$, $\Theta_{max}=30^\circ$, and the maximum transition bandwidth $4^\circ$ unless mentioned otherwise. Thus, the ideal beampattern in the sector $\xi^\star=2$ from Lemma \ref{corollary 1}, and $r_s$ is set to be $\xi^\star/10^3$ which corresponds to a $30$ dB attenuation outside the sector. $3M$ discrete points with equal space among $\mathcal{\tilde{A}}\cup\mathcal{\tilde{A}}^{-}$ are selected in P3 and P4. The optimized beampattern rather than the resulted transmit power $\rho$ is used for the performance demonstration, which is independent of the values of $R$ and ${\bar{P}}_{out}$.




In Fig. 1, the proposed scheme is compared with the ZC scheme in \cite{xinmeng1}. It can be shown that if the sector size is $180^\circ$, the ZC scheme can achieve $\xi^\star=1$ only in or near the discrete angles corresponding to the spatial frequency set $(-M/2:1:M/2-1)/(MD)$, while the proposed scheme achieves $\xi^\star=1$ for all (more than $M$) discrete angles for optimization. For the smaller sector size $60^\circ$, the beampattern for ZC scheme within the sector is still 1 while that of the proposed scheme approaches $2$ well. These improvements results from two respects. First, denser angle discretization is allowed in the proposed scheme. Second, the strict constraint of constant envelop for the ZC scheme is relaxed in the proposed scheme in which the PAPR metric is treated as an additional consideration. For the considered sector size, fair comparison of the scheme in \cite{yeli} and the proposed one is unavailable due to the lack of reference beampattern for the former. Qualitative discussions on the comparison are given  in Remark \ref{remark5}. Beamformers obtained by directly solving  P4 with a random initialization still performs worse than the proposed one even with more iterations, which validates Remark \ref{direct_P2}.
\begin{figure}[t]
\setlength{\abovecaptionskip}{-2mm} 
\centering
\includegraphics[scale=0.45]{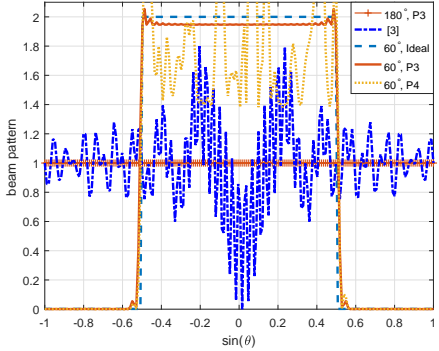}
\captionsetup{margin=5pt,font=small}
\caption{Comparison with the ZC scheme \cite{xinmeng1} and P4.}\label{Fig.4}
\vspace{-4mm}
\end{figure}

The PAPR of all beamformers with the same beampattern as that of ${\bf w}_1^\star$ are shown in Fig. 2 where $Q=8$, i.e., 256 beamformers besides ${\bf w}_1^\star$ (with beamformer index 1) are created. It can be shown that after Step 3, the PAPR of ${\bf w}_2^\star$ can be decreased by $56\%$. Further, ${\bf w}_2^\star$ can be decomposed into ${\bf W}_{RF}^\star$ and ${\bf w}_{BB}^\star$ without error by Step 4 while the details are omitted here.
\begin{figure}[t]
\setlength{\abovecaptionskip}{-2mm}
\centering
\includegraphics[scale=0.45]{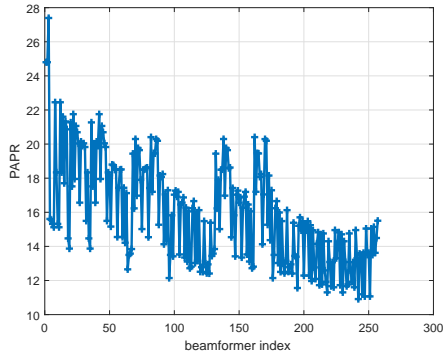}
\captionsetup{margin=5pt,font=small}
\caption{The improvement of PAPR due to Step 3.}\label{Fig.4}
\vspace{-7mm}
\end{figure}

\vspace{-5mm}
\section{Conclusion}
In this work, a beamforming design was proposed for massive MIMO public channel with any sector size which aims to minimize the transmit power while guaranteeing the QoS of all users. The ideal beampattern was first derived via Parseval Identity, based on which a non-convex but multiconvex problem was formulated which results in the full-digital beamformer with minimum gap with the ideal beampattern. In addition, the PAPR performance was improved through a search based on the same beampattern theorem. Finally, the full-digital beamformer was perfectly implemented in the hybrid structure through the triangle construction with the predefined baseband beamformer. Simulations validated the advantages of the proposed scheme over existing ones.


\vspace{-3mm}
\section*{Appendix A: Proof of Lemma \ref{corollary 1}}
\label{Appendix A}
We first give a upper bound on the sum beampattern within the sector, based on which the optimal beampattern is provided.

From the Parseval Identity \cite{Oppenheim}, we have\vspace{-0.2cm}
\begin{equation}\label{10_170324}
\small
\frac{1}{{2\pi }}{\int_{-\pi} ^{ \pi } {\left| {F\left( \Omega  \right)} \right|} ^2}d\Omega  = \sum\limits_{m = 0}^{M - 1} {{{\left| {w_m} \right|}^2}}  = {{\bf{w}}^H}{\bf{w}} = 1,\vspace{-0.2cm}
\end{equation}
\vspace{-0.2cm}where
{\small \vspace{-0.2cm}
\begin{eqnarray}\label{equivalence_170323}
  \hspace{-10pt}|F\left( \Omega  \right)|^2  &\hspace{-8pt}=\hspace{-8pt}& \left|\sum_{m = 0}^{M - 1} {w_m^{*}{e^{ - j\Omega m}}}\right|^2\hspace{-2pt}=\hspace{-2pt}\left|{\bf w}^H{{\boldsymbol \alpha}\left(\arcsin\left(\frac{\Omega}{2\pi D}\right)\right)}\right|^2 \nonumber\\
  &\hspace{-8pt}=\hspace{-8pt}& g\left(\arcsin\left(\frac{\Omega}{2\pi D}\right)\right).
\end{eqnarray}}Therefore,
{\small \vspace{-0.2cm}
\begin{eqnarray}\label{eq_14}
  \hspace{-15pt}\int_{{\sin(\Theta _{\min })}}^{{\sin(\Theta _{\max })}} {g\left( \arcsin(x)  \right)} dx \hspace{-10pt}&\mathop  = \limits^{(a)}&\hspace{-10pt} \frac{1}{{2\pi D }}{\int_{2\pi D\sin \left( {{\Theta _{\min }}} \right)}^{2\pi D\sin \left( {{\Theta _{\max }}} \right)} {\left| {F\left( \Omega  \right)} \right|} ^2}d\Omega \nonumber\\
   &\le&\hspace{-10pt} \frac{1}{{2\pi D }}{\int_{-\pi} ^{ \pi } {\left| {F\left( \Omega  \right)} \right|} ^2}d\Omega  = \frac{1}{D},\vspace{-0.2cm}
\end{eqnarray}}where (a) follows from defining  $x\triangleq\frac{\Omega}{2\pi D}$. If $F\left( \Omega  \right) =0$ for $\Omega  \in \left\{ {\left[ -\pi, 2\pi D\sin \left( {{\Theta _{\min }}} \right)\right] \cup \left[ 2\pi D\sin \left( {{\Theta _{\max }}} \right), \pi\right]} \right\}$, $\int_{{\sin(\Theta _{\min })}}^{{\sin(\Theta _{\max })}} {g\left( \arcsin(x)  \right)} dx$ approaches the upper bound $1/D$.

Secondly, we prove that $g\left( \arcsin(x)  \right)$ should be constant, i.e., $g\left( \arcsin(x)\right)=c,\forall x \in \mathcal{\tilde{A}}$. Assume the smallest beampattern is $g\left( \arcsin(x)  \right) = c_1<c$ for $x \in \left[\theta_1^{min} ,\theta_1^{max} \right]$. Since users' PASs have random interval and profile which belong to $\mathcal{F}$, if there is a User $i$ with PAS interval $\left[\theta_1^{min} ,\theta_1^{max} \right]$, its outage probability is larger than \vspace{-0.2cm}
\[{\rm{Pr}}\hspace{-0.1cm}\left\{\hspace{-0.1cm} {{{\left| {h_i^{iid}} \right|}^2} \hspace{-0.2cm}\int_{{\sin(\Theta_{min})}}^{{\sin(\Theta_{max})}}\hspace{-0.2cm} {\frac{{f}_i\left( \arcsin(x)  \right)c}{\sqrt{1-x^2}}dx }  \hspace{-0.1cm}<\hspace{-0.1cm} \frac{{{2^R} \hspace{-0.1cm}- \hspace{-0.1cm}1}}{\rho }} \hspace{-0.1cm}\right\}\hspace{-0.1cm}\]
due to \eqref{eq_17_1} and a larger $\rho$ is needed to guarantee the worst case performance. On the other hand, with $g\left( \arcsin(x)  \right)  =  c, \forall x \in \mathcal{\tilde{A}}$, users with any PAS belonging to $\mathcal{F}$ have the same outage performance, increasing $\rho$ for users with worst case performance is unnecessary.

Moreover, due to \eqref{eq_17_1} a larger $c$ can result in a smaller $\rho$ while maintaining the performance. Directly based on \eqref{eq_14}, we have $c\leq(D(\sin(\Theta _{\max })-\sin(\Theta _{\min })))^{-1}$ and the equality holds for the beampattern in \eqref{eq_15}.



\ifCLASSOPTIONcaptionsoff
  \newpage
\fi

\bibliographystyle{IEEEtran}

\end{document}